%% file: LCSS.tex
\documentclass[letterpaper,10pt,conference]{ieeeconf}  

\IEEEoverridecommandlockouts                              

\usepackage{color}
\usepackage{float}
\usepackage{multirow}
\usepackage{amsmath,amsthm}
\usepackage{amssymb,eqnarray}
\usepackage{mathabx,mathtools}
\usepackage{subcaption}
\usepackage{stmaryrd}
\usepackage{psfrag}
\usepackage{pstool}

\usepackage{booktabs} 
\usepackage{multirow}   
\usepackage{array, makecell} %
\usepackage{tikz,tikzscale,pgfplots}

\DeclareMathOperator*{\essinf}{ess\,inf}
\DeclareMathOperator*{\esssup}{ess\,sup}

\makeatletter
\newcommand{\pushright}[1]{\ifmeasuring@#1\else\omit\hfill$\equationstyle#1$\fi\ignorespaces}
\newcommand{\pushleft}[1]{\ifmeasuring@#1\else\omit$\equationstyle#1$\hfill\fi\ignorespaces}
\renewcommand{\sin}{\textrm{s}}
\renewcommand{\cos}{\textrm{c}}
\makeatother

\usepackage{multicol}
\usepackage[noend]{algorithmic}
\usepackage{algorithm2e}
\usepackage[english]{babel}
\usepackage{cite}
\usepackage{mathtools}

\floatstyle{ruled}
\newfloat{algorithm}{tbp}{loa}
\providecommand{\algorithmname}{Algorithm}
\floatname{algorithm}{\protect\algorithmname}

\newtheorem{theorem}{\protect\theoremname}
\newtheorem{proposition}{\protect\propositionname}

\theoremstyle{definition}
\newtheorem{defn}[theorem]{Definition}

\providecommand{\propositionname}{\textbf{Proposition}}
\providecommand{\remarkname}{\textbf{Remark}}
\providecommand{\theoremname}{\textbf{Theorem}}
\providecommand{\lemmaname}{Lemma}
\providecommand{\assumname}{\textbf{Assumption}}
\providecommand{\probname}{\textbf{Problem}}

\newcommand{\drew}[1]{\textcolor{blue}{#1}}

\newcommand{\newsec}[1]{\vspace{0.2cm} \noindent \textbf{#1}}

\newcommand{\defemph}[1]{\emph{\textbf{#1}}}

\title{\LARGE \textbf{Safe Control for Nonlinear Systems with Stochastic Uncertainty via \\ Risk Control Barrier Functions}}
\author{Andrew Singletary, Mohamadreza Ahmadi, and Aaron D. Ames \thanks{The authors are with the Center for Autonomous Systems and Technologies (CAST) at the California Institute of Technology, 1200 E. California Blvd., MC 104-44, Pasadena, CA 91125,  e-mail: \{asinglet, mrahmadi, ames\}@caltech.edu.}}

\pgfplotsset{compat=1.17}
\begin{document}

\maketitle

\begin{abstract}
Guaranteeing safety for robotic and autonomous systems in real-world environments is a challenging task that requires the mitigation of stochastic uncertainties.  
Control barrier functions have, in recent years, been widely used for enforcing safety related set-theoretic properties, such as forward invariance and reachability, of nonlinear dynamical systems. 
In this paper, we extend this rich framework to nonlinear discrete-time systems subject to stochastic uncertainty and propose a framework for assuring risk-sensitive safety in terms of coherent risk measures. 
To this end, we introduce risk control barrier functions (RCBFs), which are compositions of barrier functions and dynamic, coherent risk measures. We show that the existence of such barrier functions implies invariance in a coherent risk sense. Furthermore, we formulate conditions based on finite-time RCBFs to guarantee finite-time reachability to a desired set in the coherent risk. Conditions for risk-sensitive safety and finite-time reachability of sets composed of Boolean compositions of multiple RCBF are also formulated.  We show the efficacy of the proposed method through its application to a cart-pole system in a safety-critical scenario.  
\end{abstract}

\input{Sections/introduction}

\input{Sections/coherent_risk_measures}

\input{Sections/risk-sensitive_safety_and_reachability}

\input{Sections/risk_cbfs}

\input{Sections/boolean_composition}
\input{Sections/experiments}

\input{Sections/conclusions}

\footnotesize{
\bibliography{references}
}
\bibliographystyle{ieeetr}

\end{document}

%% file: Sections/introduction.tex
\section{Introduction}

Autonomous robotic systems are being increasingly deployed in real-world settings where safety is critical. With this transition to practice, the associated risk that stems from unknown and unforeseen circumstances is correspondingly on the rise~\cite{thrun2005probabilistic}. In the context of safety-critical scenarios, such as those found in aerospace and human-robot applications, it is essential that decision making accounts for risk. These risks are often associated with uncertainty due to extremely intricate nonlinear dynamics, e.g. bipedal robots~\cite{reher2020dynamic}, and/or extreme unstructured environments, e.g. subterranean or extraterrestrial exploration~\cite{rouvcek2019darpa}.

Mathematically speaking, risk can be quantified in numerous ways, such as chance constraints~\cite{ono2015chance,wang2020non}, exponential utility functions~\cite{koenig1994risk}, and distributional robustness~\cite{xu2010distributionally}. However, applications in autonomy and robotics require more ``nuanced assessments of risk''~\cite{majumdar2020should}. Artzner \textit{et. al.}~\cite{artzner1999coherent} characterized a set of natural properties that are desirable for a risk measure, called a coherent risk measure, and  have  obtained widespread
acceptance in finance and operations research, among other fields. An important example of a coherent risk measure is the conditional value-at-risk (CVaR) that has received significant attention in decision making problems, such as Markov decision processes (MDPs)~\cite{chow2015risk,chow2014algorithms,prashanth2014policy,bauerle2011markov}.  For stochastic discrete-time dynamical systems, a model predictive control technique with coherent risk objectives was proposed in~\cite{singh2018framework}, wherein the authors also proposed Lyapunov condition for risk-sensitive exponential stability. Moreover, a method based on stochastic reachability analysis was proposed in~\cite{chapman2019risk} to estimate a CVaR-safe set of initial conditions via the solution to an MDP.

Our approach to risk-sensitive safety is based on a special class of control barrier functions. Control barrier functions were proposed in~\cite{ames2016control} and have been used for designing safe controllers (in
the absence of a legacy controller, i.e., a desired controller that may be unsafe) and safety filters (in
the presence of a legacy controller) for continuous-time dynamical systems, such as bipedal robots~\cite{nguyen20163d}
and trucks~\cite{chen2019enhancing}, with guaranteed robustness~\cite{xu2015robustness,kolathaya2018input} (see the survey~\cite{ames2019control} and references therein). For discrete-time systems, discrete-time barrier functions were formulated in~\cite{ahmadi2019safe,agrawal2017discrete} and applied to the multi-robot coordination problem~\cite{ahmadi2020barrier}. Recently, for a class of stochastic (Ito) differential equations, safety in probability and statistical mean was studied in~\cite{clark2019control,santoyo2019barrier} via stochastic barrier functions.



\begin{figure}[t] \centering{
\includegraphics[width=1\columnwidth]{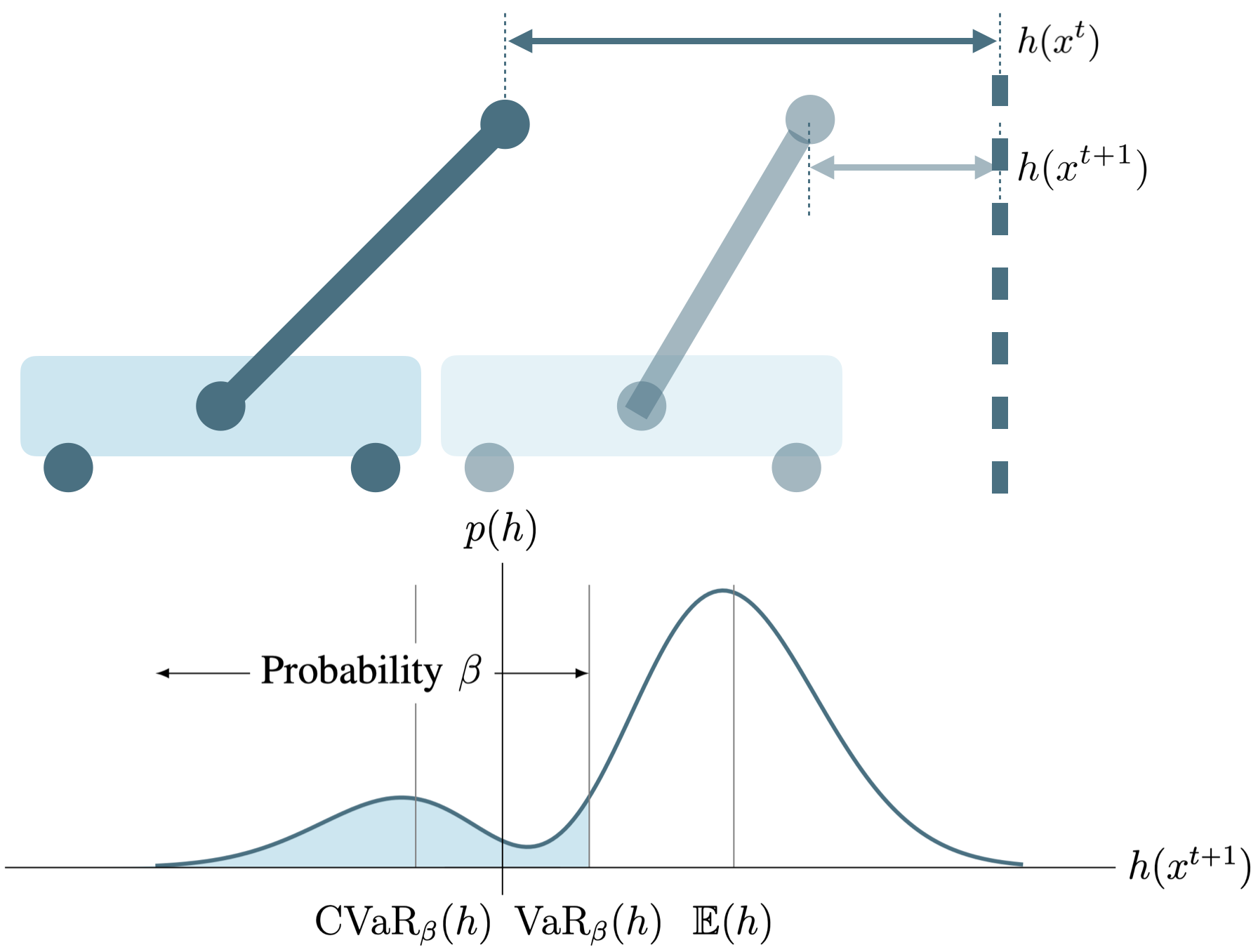}
\caption{The value of the safe-set $h(x^t)$ is known at time $t$, but stochastic uncertainty makes $h(x^{t+1})$ a random variable. We must pick $u^t$ such that $h(x^{t+1})$ is safe subject to a risk measure taken over the worst $\beta$ probability.}\label{fig:introfig}
}
 \end{figure}

This paper goes beyond the conventional notions of safety in probability and statistical mean through the use of coherent risk measures (as motivated in Section \ref{sec:coherent}). To this end, in Section \ref{sec:safetyreachability}, for discrete-time systems subject to stochastic uncertainty, we define safety and finite-time reachability in the risk-sensitive sense, i.e., in the context of the worst possible realizations, via coherent risk measures.  We then propose \emph{risk control barrier functions} (RCBFs) in Section \ref{sec:RCBF}, together with finite-time RCBFs, as a tool to enforce risk-sensitive safety and reachability, respectively.  The main result of this paper establishes that RCBFs ensure safety in a risk sensitive fashion.  Finite-time RCBFs allow for the extension of this result to risk-sensitive reachability.  
Furthermore, for safe and goal sets defined as Boolean compositions of multiple function level-sets, we propose conditions that ensure safety and reachability of these sets based on RCBFs and their finite-time counterparts.
Importantly, in all cases, the risk-sensitive controllers are designed to minimally invasive with respect to a given system legacy controller. We show the efficacy of our approach in Section \ref{sec:experiments} through simulation on a nonlinear cart-pole system (see Figure~\ref{fig:introfig}).


\vspace{0.2cm}
\noindent \textbf{Notation: } We denote by $\mathbb{R}^n$ the $n$-dimensional Euclidean space and $\mathbb{N}_{\ge0}$ the set of non-negative integers.  For a finite set $\mathcal{A}$, we denote by $|\mathcal{A}|$ the number of elements of $\mathcal{A}$. For  a probability space $(\mathcal{X}, \mathcal{F}, \mathbb{P})$ and a constant $p \in [1,\infty)$, $\mathcal{L}_p(\mathcal{X}, \mathcal{F}, \mathbb{P})$ denotes the vector space of real valued random variables $X$ for which $\mathbb{E}|X|^p < \infty$. The Boolean operators are denoted by $\neg$ (negation), $\lor$ (conjunction), and $\land$ (disjunction). For a risk measure or a function $\rho$, we denote $\rho^t$ to show the function composition of $\rho$ with itself $t$ times.

%% file: Sections/coherent_risk_measures.tex
\section{Coherent Risk Measures}
\label{sec:coherent}

The goal of this section is to introduce conditional risk measures with a view toward defining risk control barrier functions in subsequent sections.  
In this context, consider a probability space $(\Omega, \mathcal{F}, \mathbb{P})$, a filtration $\mathcal{F}_0 \subset \cdots \mathcal{F}_N \subset \mathcal{F} $, and an adapted sequence of random variables $h^t,~t=0,\ldots, N$, where $N \in \mathbb{N}_{\ge 0} \cup \{\infty\}$.
For $t=0,\ldots,N$, we further define the spaces $\mathcal{H}_t = \mathcal{L}_p(\Omega, \mathcal{F}_t, \mathbb{P})$, $p \in [0,\infty)$,  $\mathcal{H}_{t:N}=\mathcal{Z}_t\times \cdots \times \mathcal{Z}_N$ and $\mathcal{H}=\mathcal{H}_0\times \mathcal{H}_1 \times \cdots$. We  assume that the sequence $\boldsymbol{h} \in \mathcal{H}$ is almost surely bounded (with exceptions having probability zero), \textit{i.e.}, 
 $
\esssup_t~| h^t(\omega) | < \infty.
$
In order to describe how one can evaluate the risk of sub-sequence $h_t,\ldots, h_N$ from the perspective of stage $t$, we require the following definitions.




\begin{defn}[Conditional Risk Measure]
A mapping $\rho_{t:N}: \mathcal{H}_{t:N} \to \mathcal{H}_{t}$, where $0\le t\le N$, is called a \defemph{conditional risk measure}, if it has the following monotonicity property:
\begin{equation*}
    \rho_{t:N}(\boldsymbol{h}) \le   \rho_{t:N}(\boldsymbol{h}'), \quad \forall \boldsymbol{h}, \forall \boldsymbol{h}' \in \mathcal{H}_{t:N}~\text{such that}~\boldsymbol{h} \preceq \boldsymbol{h}'.
\end{equation*}
A \defemph{dynamic risk measure} is a sequence of conditional risk measures $\rho_{t:N}:\mathcal{H}_{t:N}\to \mathcal{H}_{t}$, $t=0,\ldots,N$.
\end{defn}

One fundamental property of dynamic risk measures is their consistency over time~\cite[Definition 3]{ruszczynski2010risk}. That is, if $h$ will be as good as $h'$ from the perspective of some future time $\theta$, and they are identical between times $\tau$ and $\theta$, then $h$ should not be worse than $h'$ from the perspective at time $\tau$.
If a risk measure is time-consistent, we can define the one-step conditional risk measure $\rho_t:\mathcal{H}_{t}\to \mathcal{H}_{t-1}$, $t=0,\ldots,N-1$ as follows:
\begin{equation}
    \rho_t(h^{t}) = \rho_{t-1,t}(0,h^{t}),
\end{equation}
and for all $t=1,\ldots,N$, we obtain:
\begin{multline}
    \label{eq:dynriskmeasure}
    \rho_{t,N}(h^t,\ldots,h^N)= \rho_t \big(h^t + \rho_{t+1} ( h^{t+1}+\rho_{t+2}(h^{t+2}+\cdots\\
    +\rho_{N-1}\left(h^{N-1}+\rho_{N}(h^N) \right) \cdots )) \big).
\end{multline}
Note that the time-consistent risk measure is completely defined by one-step conditional risk measures $\rho_t$, $t=0,\ldots,N-1$ and, in particular, for $t=0$, \eqref{eq:dynriskmeasure} defines a risk measure of the entire sequence $\boldsymbol{h} \in \mathcal{H}_{0:N}$.
This leads to the notion of a coherent risk measure. 


\begin{defn}[Coherent Risk Measure]\label{defi:coherent}
We call the one-step conditional risk measures $\rho_t: \mathcal{H}_{t+1}\to \mathcal{H}_t$, $t=1,\ldots,N-1$ as in~\eqref{eq:dynriskmeasure} a \defemph{coherent risk measure} if it satisfies the following conditions
\begin{itemize}
    \item \textbf{Convexity:} $\rho_t(\lambda h + (1-\lambda)h') \le \lambda \rho_t(h)+(1-\lambda)\rho_t(h')$, for all $\lambda \in (0,1)$ and all $h,h' \in \mathcal{H}_{t}$;
    \item \textbf{Monotonicity:} If $h\le h'$ then $\rho_t(h) \le \rho_t(h')$ for all $h,h' \in \mathcal{H}_{t}$;
    \item \textbf{Translational Invariance:} $\rho_t(h+h')=c+\rho_t(h')$ for all $h \in \mathcal{H}_{t-1}$ and $h' \in \mathcal{H}_{t}$;
    \item \textbf{Positive Homogeneity:} $\rho_t(\beta h)= \beta \rho_t(h)$ for all $h \in \mathcal{H}_{t}$ and $\beta \ge 0$.
\end{itemize}
\end{defn}

All risk measures studied in this paper are time-consistent coherent risk measures.
Concretely, we briefly review two examples of coherent risk measures. 




\newsec{Total Conditional Expectation:} The simplest risk measure is the total conditional expectation given by
\begin{equation}
    \rho_t(h^{t}) =  \mathbb{E}\left[  h^{t} \mid \mathcal{F}_{t-1}    \right].
\end{equation} 
It is easy to see that total conditional expectation satisfies the properties of a coherent risk measure as outlined in Definition~\ref{defi:coherent}. Unfortunately, total conditional expectation is agnostic to realization fluctuations of the stochastic variable $h$ and is only concerned with the mean value of $h$ at large number of realizations. Thus, it is a risk-neutral measure of performance.

\newsec{Conditional Value-at-Risk:} Let $h \in \mathcal{H}$ be a stochastic variable for which higher values are of interest\footnote{For example, greater values of $h$ indicate safer performance as will be discussed in the next section.}. For a given confidence level $\beta \in (0,1)$, value-at-risk ($\mathrm{VaR}_{\beta}$) denotes the $\beta$-quantile value of a stochastic  variable $h \in \mathcal{H}$ described as 
$
\mathrm{VaR}_{\beta}(h) = \sup_{\zeta \in \mathbb{R}} \{ \zeta \mid \mathbb{P}(h \le \zeta ) \le \beta \}.
$ 
Unfortunately, working with VaR  for non-normal stochastic variables is numerically unstable, optimizing models involving  VaR are intractable in
high dimensions, and VaR ignores the values of $h$ with probability less than $\beta$~\cite{rockafellar2000optimization}. 

In contrast, CVaR overcomes the shortcomings of VaR. CVaR with confidence level $\beta \in (0,1)$ denoted $\mathrm{CVaR}_{\beta}$ measures the expected loss in the $\beta$-tail given that the particular threshold $\mathrm{VaR}_{\beta}$ has been crossed, i.e., $\mathrm{CVaR}_{\beta} (h) =  \mathbb{E}\left[ h \mid h \le \mathrm{VaR}_{\beta}(h)  \right]$. An optimization formulation for CVaR was proposed in~\cite{rockafellar2000optimization} that we use in this paper. That is, $\mathrm{CVaR}_{\beta}$ is given by 
\begin{align}
    \mathrm{CVaR}_{\beta}(h):&=-\inf_{\zeta \in \mathbb{R}}\mathbb{E}\left[\zeta + \frac{(-h-\zeta)_{+}}{\beta}\right]. \label{eq:cvardual}
\end{align}
Note that the above formulation of CVaR is concerned with the left-tail of distributions (higher values of $h$ are preferred).

A value of $\beta \to 1$ corresponds to a risk-neutral case, i.e.,  $\mathrm{CVaR_1}(h)=\mathbb{E}(h)$; whereas, a value of $\beta \to 0$ is rather a risk-averse case, i.e., $\mathrm{CVaR_0}(h)=\mathrm{VaR}_0(h)= \essinf(h)$~\cite{rockafellar2002conditional}. Figure~\ref{fig:introfig} illustrates these notions for an example $h$ variable with distribution $p(h)$.

%% file: Sections/risk-sensitive_safety_and_reachability.tex
\section{Risk-Sensitive Safety and Reachability}
\label{sec:safetyreachability}

We assume the robot dynamics of interest is described by a discrete-time stochastic system  given by
\begin{equation}\label{eq:dynamics}
    x^{t+1} = f(x^t,u^t,w^t), \quad x^0=x_0,
\end{equation}
where $t \in \mathbb{N}_{\ge 0}$ denotes the time index, $x \in \mathcal{X} \subset \mathbb{R}^n$ is the state, $u \in \mathcal{U} \subset \mathbb{R}^m$ is the control input,  $w \in \mathcal{W}$ is the stochastic uncertainty/disturbance, and the function $f:  \mathbb{R}^n \times \mathcal{U} \times \mathcal{W} \to \mathbb{R}^n$. We assume that the initial condition $x_0$ is deterministic and that $|\mathcal{W}|$ is finite, \textit{i.e.,} $\mathcal{W} = \{w_1, \ldots, w_{|\mathcal{W}|}\}$. At every time-step $t$,
for a state-control pair $(x^t, u^t)$, the process disturbance $w^t$ is
drawn from set $\mathcal{W}$ according to the probability mass function $p(w) = [p(w_1),\ldots,p(w_{|\mathcal{W}|})]^T$, where $p(w_i):=\mathbb{P}(w^t=w_i)$, $i=1,2,\ldots,|\mathcal{W}|$. Note that the
probability mass function for the process disturbance is time-invariant, and that the process disturbance is independent of
the process history and of the state-control pair $(x^t, u^t)$.

Note that, in particular, system~\eqref{eq:dynamics} can capture stochastic hybrid systems, such as Markovian Jump Systems~\cite{zhao2019brief}.

We are interested in studying the properties of the solutions to~\eqref{eq:dynamics} with respect to the compact set $\mathcal{S}$ described by: 
\begin{align}
\label{eq:safeset}
\mathcal{S} :=\{ x \in \mathcal{X} \mid h(x) \ge 0 \}, \nonumber\\
\mathrm{Int}(\mathcal{S}) :=\{ x \in \mathcal{X} \mid h(x) > 0 \}, \\
\partial \mathcal{S} :=\{ x \in \mathcal{X} \mid h(x) = 0 \}, \nonumber
\end{align}
where $h:\mathcal{X} \to \mathbb{R}$ is a continuous function.

In the presence of stochastic uncertainty $w$, assuring almost sure (with probability one) invariance or safety may not be feasible. Moreover, enforcing safety in expectation is only meaningful if the law of large numbers  can  be  invoked  and  we  are  interested  in  the  long term  performance,  independent  of  the  realization  fluctuations. In this work, instead, we propose safety in the dynamic coherent risk measure sense with conditional expectation as an special case.

\begin{defn} [$\rho$-Safety]
Given a safe set $\mathcal{S}$ as given in~\eqref{eq:safeset} and a time-consistent, dynamic coherent risk measure $\rho_{0:t}$ as described in~\eqref{eq:dynriskmeasure}, we call the solutions to~\eqref{eq:dynamics}, starting at $x_0 \in \mathcal{S}$, \defemph{$\rho$-safe} if and only if 
\begin{align}\label{eq:risksafety}
    \rho_{0,t}\left(0,0,\ldots, h(x) \right) \ge 0, \quad \forall t \in \mathbb{N}_{\ge 0}.
\end{align}
\end{defn}

In order to understand \eqref{eq:risksafety}, consider the case where $\rho$ is the conventional total expectation. Then, \eqref{eq:risksafety} implies safety in expectation. As mentioned earlier, the definition of safety for general coherent risk measures goes beyond the traditional total expectation. 


Another interesting property we study in this paper arises when $x_0 \in \mathcal{X}\setminus \mathcal{S}$. That is, when instead of safety, we are interested in reaching a set of interest in finite time.

\begin{defn}[$\rho$-Reachability]
Consider system~\eqref{eq:dynamics} with initial condition $x_0 \in \mathcal{X} \setminus \mathcal{S}$. Given a set $\mathcal{S}$ as given in~\eqref{eq:safeset} and a time-consistent, dynamic coherent risk measure $\rho_{0:t}$ as described in~\eqref{eq:dynriskmeasure}, we call the set $\mathcal{S}$ \defemph{$\rho$-reachable},  if and only if there exists a constant $t^*$ such that
\begin{equation}\label{eq:ftrisksafety}
    \rho_{0,t^*}\left(0,0,\ldots, h(x) \right) \ge 0.
\end{equation}
\end{defn}

%% file: Sections/risk_cbfs.tex
\section{Risk Control Barrier Functions}
\label{sec:RCBF}

In order to check and enforce risk sensitive safety, i.e., $\rho$-safety, we introduce \emph{risk control barrier functions}.  We then extend these to a finite-time variation, which allows us to establish risk-sensitive reachability, i.e., $\rho$-reachability.

\subsection{Risk Sensitive Safety with RCBFs}

\begin{defn}[Risk Control Barrier Function]\label{def:riskbf}
For~the discrete-time system~\eqref{eq:dynamics} and a dynamic coherent risk measure $\rho$, the continuous function $h : \mathbb{R}^n \to \mathbb{R}$ is a \defemph{risk control barrier
function} for the set $\mathcal{S}$ as defined in~\eqref{eq:safeset}, if there exists a convex $\alpha \in \mathcal{K}$ satisfying $\alpha(r) < r$ for all $r>0$  such that
\begin{equation}\label{eq:BFinequality}
   \rho( h(x^{t+1})) \ge  \alpha( h(x^{t})),\quad \forall x^t \in \mathcal{X}.
    \end{equation}
    \end{defn}

Note that a simple choice for the function $\alpha$ is $\alpha=\alpha_0$, where $\alpha_0 \in (0,1)$ is a constant.

In the first main contribution of the paper, we demonstrate that the existence of an RCBF implies invariance/safety in the coherent risk measure.

\begin{theorem}\label{thm:riskbf}
Consider the discrete-time system~\eqref{eq:dynamics} and the set $\mathcal{S}$ as described in~\eqref{eq:safeset}. Let $\rho$ be a given coherent risk measure. Then, $\mathcal{S}$ is $\rho$-safe if there exists an RCBF as defined in Definition~\ref{def:riskbf}.
\end{theorem}

\begin{proof}
The proof is carried out by induction and using the properties of a coherent risk measure as outlined in Definition~\ref{defi:coherent}. If~\eqref{eq:BFinequality} holds, for $t=0$, we have 
\begin{equation}\label{ewqwq}
\rho(h(x^1))\ge \alpha ( h(x_0)).
\end{equation}
Similarly, for $t=1$, we have 
\begin{equation}\label{sdsdd}
\rho(h(x^2))\ge \alpha ( h(x\drew{^1})).
\end{equation}
Since $\rho$ is monotone, composing both sides of~\eqref{sdsdd} with $\rho$ does not change the inequality and we obtain
\begin{equation} \label{eq:bnndds}
\rho \circ \rho(h(x^2))\ge  \rho (\alpha( h(x^1))).
\end{equation}
Since $\alpha$ is a convex function, from Theorem 3 in~\cite{chen2013risk} (Jensen's Inequality for coherent risk measures), we obtain\footnote{In particular, if $\alpha \in (0,1)$ is a constant, from positive homogeneity property of $\rho$, we have
\begin{equation*}
\rho \circ \rho(h(x^2))\ge  \rho (\alpha h(x^1))=\alpha \rho ( h(x^1)).
\end{equation*}} 
\begin{equation*}
\rho \circ \rho(h(x^2))\ge  \rho (\alpha( h(x^1)))\ge \alpha( \rho ( h(x^1))).
\end{equation*}

Then, using inequality~\eqref{ewqwq}, we have 
\begin{equation*}\label{sdsdd2}
\rho \circ \rho(h(x^2))\ge  \alpha ( \rho (h(x^1))) \ge \alpha \circ \alpha( h(x_0)).
\end{equation*}
Therefore, by induction, at time $t$, we can show that
 $
 \rho^{t} (h(x^t)) \ge {\alpha}^{t} (h(x_0)).
$ 
The left-hand side of the above inequality is equal to $\rho_{0,t}(0,\ldots,h(x^t))$. Hence, 
\begin{equation}\label{eq:expdecay}
\rho_{0,t}(0,\ldots,h(x^t))\ge {\alpha}^{t}( h(x_0)).
\end{equation}
If $x_0 \in \mathcal{S}$, from the definition of the set $\mathcal{S}$, we have $h(x_0)\ge0$. Since $\alpha \in \mathcal{K}$, then we can infer that~\eqref{eq:risksafety} holds. Thus, the system is $\rho$-safe.
\end{proof}
\vspace{0.2cm}


Note that, in the case when $x_0\in \mathcal{X} \setminus \mathcal{S}$, the existence of an RCBF implies asymptotic convergence to the set $\mathcal{S}$ in the coherent risk measure $\rho$. This can be inferred from~\eqref{eq:expdecay}.  In fact, if
$\alpha(r) < r$, then there exist a constant $\delta \in (0, 1)$ such
that $\alpha (r) \le \delta r$ and hence \begin{equation}\label{sdscs} {\alpha}^{t}(r) \le \delta^t r, \quad t \in \mathbb{N}_{\ge0},\end{equation}
If $x_0\in \mathcal{X} \setminus \mathcal{S}$,
then $h(x_0)<0$. However, from~\eqref{sdscs}, as $t\to \infty$, $\alpha \circ \cdots \circ \alpha(r) \to 0$, since the compositions of class $\mathcal{K}$ functions is also class $\kappa$ (hence non-negative). We then obtain $\rho_{0,t}(0,\ldots,h(x^t))\ge 0$, which implies that the solutions become $\rho$-safe.

\subsection{Risk Sensitive Safety with Finite-time RCBFs}

In practice, we are often interested in satisfying system specifications characterized by the set $\mathcal{S}$ in finite time. To this end, we define finite-time RCBFs.

\begin{defn}[Finite-Time RCBF] \label{def:ftdtbf}
For the discrete-time system~\eqref{eq:dynamics} and a dynamic coherent risk measure $\rho$, the  continuous function ${h}:\mathcal{X} \to \mathbb{R}$ is a \defemph{finite-time RCBF} for the set $\mathcal{S}$ as defined in~\eqref{eq:safeset}, if there exist constants  $0<\gamma<1$ and $\varepsilon > 0$ such that
\begin{equation}\label{eq:BFft}
    \rho({h}(x^{t+1}))-\gamma {h}(x^{t}) \ge \varepsilon (1-\gamma),\quad \forall {  x^t \in \mathcal{X}}.
\end{equation}
\end{defn}


In the second key contribution of the paper, we show that the existence of a finite-time RCBF implies $\rho$-reachability.

\begin{theorem} \label{thm:FTDTBF}
Consider the discrete-time system~\eqref{eq:dynamics} and a dynamic coherent risk measure $\rho$. Let $\mathcal{S} \subset \mathcal{X}$ be as described in~\eqref{eq:safeset}. If there exists a finite-time RCBF ${h}:\mathcal{X}\to \mathbb{R}$ as in Definition~\ref{def:ftdtbf}, then for all $x^0 \in \mathcal{X}\setminus \mathcal{S}$, there exists a $t^* \in \mathbb{N}_{\ge 0}$ such that $\mathcal{S}$ is $\rho$-reachable, \textit{i.e.,} inequality~\eqref{eq:ftrisksafety} holds. Furthermore, 
\begin{equation} \label{eq:upperboundtheorem2}
    t^* \le {\log\left(\frac{\varepsilon - {h}(x^0)}{\varepsilon}\right)}/{\log\left(\frac{1}{\gamma}\right)},
    \end{equation}
    where the constants $\gamma$ and $\varepsilon$ are as defined in Definition~\ref{def:ftdtbf}.
\end{theorem}

\begin{proof}
Similar to the proof of Theorem 1, we use induction and properties of coherent risk measures. We prove by induction. From~\eqref{eq:BFft}, we have 
 $
\rho({h}(x^{t+1}))-\varepsilon \ge \gamma {h}(x^t) - \gamma \varepsilon = \gamma \left(  {h}(x^t) -    \varepsilon      \right).
$ 
Hence, for $t=0$, we have  
\begin{equation}\label{xcwsaaa}
\rho(h(x^1))-\varepsilon \ge \gamma (h(x_0)-\varepsilon).
\end{equation}
For $t=1$, we have 
\begin{equation}
\rho(h(x^2))-\varepsilon \ge \gamma (h(x^1)-\varepsilon).
\end{equation}
Since $\rho$ is monotone, composing both sides of the above inequality with $\rho$ does not change the inequality and we obtain
$$
\rho \circ \rho (h(x^2)-\varepsilon) \ge \rho (\gamma (h(x^1)-\varepsilon)) = \gamma \rho  (h(x^1)-\varepsilon),
$$
where in the last equality we used the positive homogeneity property of $\rho$ since $\gamma \in (0,1)$. Since $\varepsilon>0$ is a constant, translational invariance property of $\rho$ yields
$$
\rho \circ \rho (h(x^2))-\varepsilon \ge  \gamma (\rho  (h(x^1))  -\varepsilon).
$$
Moreover, from inequality~\eqref{xcwsaaa}, we infer
$$
\rho \circ \rho (h(x^2))-\varepsilon \ge  \gamma (\rho  (h(x^1))  -\varepsilon) \ge \gamma^2 (h(x_0)  -\varepsilon).
$$
Thus, by induction, we see that at time step $t$, the following inequality holds
\begin{equation*}
\rho^{t} (h(x^t))-\varepsilon \ge \gamma^t (h(x_0)  -\varepsilon).
\end{equation*}
Taking $\varepsilon$ to the right-hand side and noting that the left-hand side of the above inequality is equal to $\rho_{0,t}(0,\ldots,h(x^t))$, we have the following inequality 
\begin{equation} \label{cxcxvcvcvdss}
\rho_{0,t}(0,\ldots,h(x^t)) \ge \gamma^t (h(x_0)  -\varepsilon) + \varepsilon.
\end{equation}

 Since $0<\gamma<1$ and $x^0 \in \mathcal{X}\setminus \mathcal{S}$, \textit{i.e.,} ${h}(x^0)<0$, as $t$ increases $x^t$ approaches $\mathcal{S}$ in the dynamic risk measure $\rho_{0,t}$, because by definition $h(x^t)\ge 0$ implies $x^t \in \mathcal{S}$. Hence,  $\mathcal{S}$ is $\rho$-reachable in finite time.
 
 by definition, $x^t$ reaches $\mathcal{S}$ at least at the boundary by $t^*$ when $\tilde{h}(x^t)=0$. Substituting $\tilde{h}(x^t)=0$ in~\eqref{cxcxvcvcvdss} yields
 \begin{equation} \label{cxcxvcvcvdss2}
0 \ge \gamma^{t^*} (h(x_0)  -\varepsilon) + \varepsilon,
\end{equation}
 where we used the fact that $\rho_{0,t}(0,\ldots,h(x^{t^*}))=\rho_{0,t}(0,\ldots,0)=0$. Re-arranging the term and noting that $h(x_0)\le 0$ and therefore $h(x_0)-\varepsilon \le 0$, we obtain
 $$
 \frac{\varepsilon }{\varepsilon - {h}(x_0)} \le \gamma^t.
 $$ 
 Taking the logarithm of both sides of the above inequality  gives
 $
\log\left(\frac{ \varepsilon}{\varepsilon -{h}(x_0)     }\right) \le t\log(\gamma),
$ 
or equivalently 
$$
-\log\left(  \frac{\varepsilon-{h}(x_0) }{    \varepsilon}\right) \le -t \log(\frac{1}{\gamma}).
$$
  Since $0<\gamma<1$, $\log(\frac{1}{\gamma})$ is a positive number. Dividing both sides of the inequality above with the negative number $-\log(\frac{1}{\gamma})$ obtains 
$
  t \le {\log\left(\frac{\varepsilon - \tilde{h}(b^0)}{\varepsilon}\right)}/{\log\left(\frac{1}{\rho}\right)}.
$  
\end{proof}
\vspace{0.2cm}

The upper bound described by inequality~\eqref{eq:upperboundtheorem2} in Theorem 2 is dependent on the two parameter $\gamma$ and $\varepsilon$. In our experiments, we often fix $0<\gamma<1$ and carry out a line search over $\varepsilon$ until the finite-time RCBF condition~\eqref{eq:BFft} does not hold anymore. Then, we pick the corresponding $t^*$ as the upper-bound on the earliest time the solutions can enter the goal set $\mathcal{S}$.

\begin{figure*}
    \centering
    \begin{subfigure}{.99\textwidth}
    \includegraphics[trim={7.5cm 0 7.5cm 0},clip,width=\columnwidth]{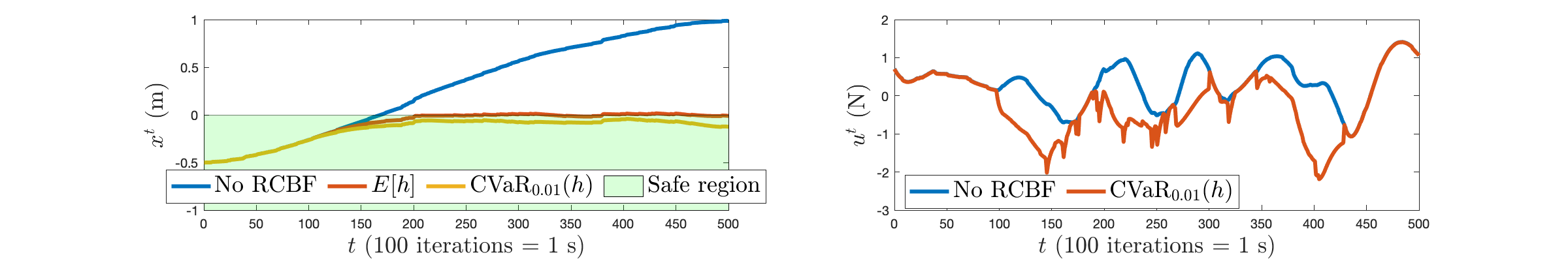}
    \end{subfigure}
    \begin{subfigure}{.99\textwidth}
    \includegraphics[trim={7.5cm  0 7.5cm 0},clip,width=\columnwidth]{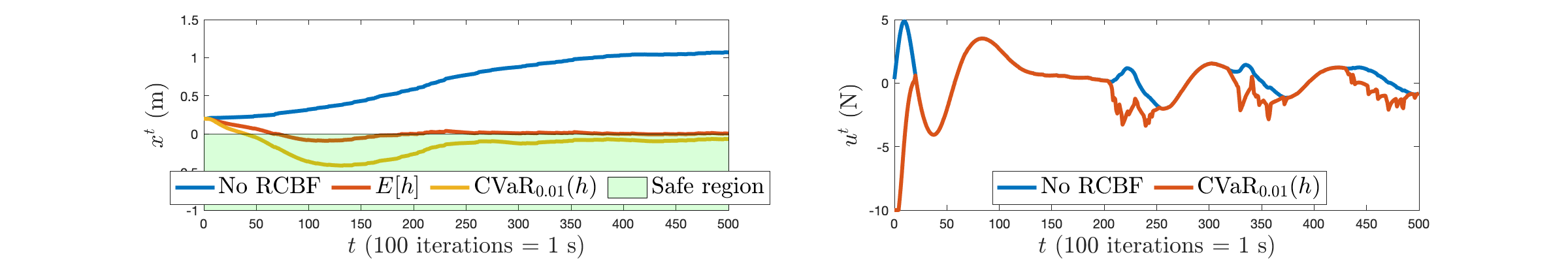}
    \end{subfigure}
    \caption{Simulation results for the cart-pole system with no RCBF filter, and with standard RCBF (top) and finite-time RCBF (bottom) filters using total conditional expectation and CVaR. }\label{fig:simfig}
\end{figure*}

%% file: Sections/boolean_composition.tex
\subsection{Boolean Compositions of RCBFs}
\label{sec:boolean}

We have proposed RCBFs and finite-time RCBFs as means to verify $\rho$-safety and $\rho$-reachability, respectively. 
We now propose conditions for verifying $\rho$-safety and $\rho$-reachability for Boolean compositions of several control barrier functions
~\cite{glotfelter2017nonsmooth,ahmadi2020barrier,ahmadi2020safe}. 
We omit proofs due to space constraints.
\begin{proposition} \label{prop:booleanBFs}
Let $\mathcal{S}_i = \{ x \in \mathbb{R}^n \mid h_i(x)\ge 0\}$, $i=1,\ldots,k$ denote a family of safe sets with the boundaries and interior defined analogous to $\mathcal{S}$ in~\eqref{eq:safeset} and $\rho$ be a given dynamic coherent risk measure. Consider the discrete-time system~\eqref{eq:dynamics}. If there exist a $\alpha \in (0,1)$ such that
\begin{equation}\label{eq:disjBF1}
 \rho\left(\min_{i=1,\ldots,k} h_i(x^{t+1})\right)  \ge  \alpha\min_{i=1,\ldots,k}  h_i(x^{t})
\end{equation}
then the set $\{ x \in \mathbb{R}^n \mid  \land_{i =1,\ldots,k} \left(h_i(x) \ge 0\right)\}$ is $\rho$-safe. Similarly, if there exist a $\alpha \in (0,1)$  such that
\begin{equation}\label{eq:disjBF2}
 \rho\left(\max_{i=1,\ldots,k} h_i(x^{t+1})\right) \ge  \alpha\max_{i=1,\ldots,k}  h_i(x^{t})
\end{equation}
then the set $\{ x  \in \mathbb{R}^n \mid  \lor_{i =1,\ldots,k} \left(h_i(x) \ge 0\right)\}$ is $\rho$-safe.
\end{proposition}
We next propose conditions for risk-sensitive finite-time reachability of sets composed of Boolean compositions of several functions $h$ as described in~\eqref{eq:safeset}.

\begin{proposition} \label{prop:booleanftBFs}
Let $\mathcal{S}_i = \{ x \in \mathbb{R}^n \mid h_i(x)\ge 0\}$, $i=1,\ldots,k$ denote a family of sets with the boundaries and interior defined analogous to $\mathcal{S}$ in~\eqref{eq:safeset} and $\rho$ be a given dynamic coherent risk measure. Consider the discrete-time system~\eqref{eq:dynamics}. If there exist constants  $0<\gamma<1$ and $\varepsilon > 0$ such that
\begin{equation}\label{eq:disjtbBF1}
 \rho\left(\min_{i=1,\ldots,k} h_i(x^{t+1})\right) -\gamma \min_{i=1,\ldots,k} h_i(x^{t}) \ge  \varepsilon (1-\gamma)
\end{equation}
then the set $\{ x \in \mathbb{R}^n \mid  \land_{i =1,\ldots,k} \left(h_i(x) \ge 0\right)\}$ is $\rho$-reachable. Then, there exists a constant $t^*$ satisfying
\begin{equation}\label{eq:disjtbBF1time}
        t^* \le {\log\left(\frac{\varepsilon - \min_{i=1,\ldots,k}{h}_i(x^0)}{\varepsilon}\right)}/{\log\left(\frac{1}{\gamma}\right)}, 
\end{equation}
such that if $x^0 \in \mathcal{X} \setminus \cup_{i=1,\ldots,k} \mathcal{S}_i$ then $x^{t^*} \in \cap_{i=1,\ldots,k} \mathcal{S}_i$.
Similarly, the disjunction
case follows by replacing $\min$ with $\max$ in \eqref{eq:disjtbBF1} and \eqref{eq:disjtbBF1time}.
\end{proposition}

%% file: Sections/experiments.tex
\section{Simulation Results} \label{sec:experiments}


In order to illustrate the results of these risk-aware guarantees, we apply our method in the case of the cart-pole, modeled as a nonlinear, control-affine discrete-time system. 
\begin{equation}
x^{t+1} = x^t +
\begin{bmatrix}
v_x \\
\dot{\theta} \\
\frac{u^t + m_p \sin\theta (l \dot\theta^2 + g\cos\theta)}{m_c + m_p \sin^2\theta}\\
\frac{ -u^t
       \cos\theta - m_p l \dot\theta^2 \cos\theta \sin\theta - (m_c + m_p) g \sin\theta }{l(m_c + m_p \sin^2\theta)} 
\end{bmatrix}
\Delta_t + w^t
\end{equation}
The disturbance $w^t \in W$ enters the system linearly, and is described by a pmf over the states. This could include the modeling error from this Euler-approximated discrete-time model, but in this case, it is a simple pmf normally distributed around $0$ with standard deviation $\sigma = \{0.05,0.05, 0.2, 0.2\}$ for the four states $x = \left[ p_x,\theta, v_x,\dot{\theta}\right]$.

The safety set is described by
\begin{equation}
    h(x^t) = -2a_{\max}p_x^t - {v_x^t}^2\textrm{sgn}(v_x^t),
\end{equation}
where $a_{\max} > 0$ is a tuneable parameter that designates the maximum linear acceleration at any point. This function is positive when $p_x < 0$, but allows $h(x^t) > 0$ when $p_x > 0$ if $v_x$ is sufficiently negative.

While this safety set is nonlinear in the control inputs, the one-step nature of this optimization problem results in no issues solving such a program in real-time, using modern solvers such as IPOPT or NLOPT. In future work, we plan to show how nonlinear CBFs can be linearized to result in an affine RCBF constraint, with the error included in the stochastic uncertainty to result in formal safety guarantees.

The RCBF was solved using PAGMO's integrated SLSQP solver from NLOPT. Each solution took roughly 0.7 ms to compute on a modern laptop, resulting in a maximum control frequency of 1428 Hz. Three trajectories are shown in Figure \ref{fig:simfig}. The desired trajectory shows the trajectory with only the nominal controller, which clearly surpasses the safe set at $x = 0$. The trajectory corresponding to $\mathbb{E}[h]$ was filtered subject to the total conditional expectation coherent risk measure, which also corresponds to CVaR with $\beta = 1$. While this filter guarantees safety in the expectation, safety is frequently violated due to the stochastic uncertainty. Finally, the trajectory corresponding to CVaR with $\beta = 0.01$ results in safety over the entire trajectory.

Similarly, Figure \ref{fig:simfig} also demonstrates the same three trajectories with the finite-time reachability RCBF. Specifically, we utilize constants $\gamma = 0.05$ and $\epsilon = 0.1$, with an initial safety violation of $h(x^0) = -0.2$. From \eqref{eq:upperboundtheorem2}, this suggests a $t^*\leq 0.3667$s. While this is not reflected in the plot, which only shows $p_x^t$ rather that $h(x^t)$, we find that $h(x^{t^*}) > 0$ at $t^* = 0.08$s, well below the theoretical guarantee. 

%% file: Sections/conclusions.tex
\section{Conclusions}

In this paper, we propose Risk Control Barrier Functions (RCBFs) as a means to enforce safety in the presence of stochastic uncertainty. We demonstrate theoretically that these RCBFs guarantee safety with respect to dynamic coherent risk measures, which serve as a computationally efficient means to assess risk. Moreover, we proved that finite-time RCBFs can be utilized to guarantee convergence to a set in finite time, resulting in a practical safety filter that works both inside and outside of the safe set. We also demonstrated how multiple safe sets can be enforced simultaneously utilizing Boolean compositions. Finally, we demonstrated the efficacy of this framework on the nonlinear cart-pole system under stochastic uncertainty. 